\documentclass[a4paper,UKenglish]{lipics-v2016}

\usepackage[utf8]{inputenc}
\usepackage{mathscinet}
\usepackage{comment}
\usepackage{graphicx}
\usepackage{booktabs}
\usepackage{hyperref}
\usepackage{pgfkeys}
\usepackage[numbers,sort&compress]{natbib}
\usepackage{enumerate}
\usepackage{verbatim}
\usepackage{nicefrac}
\usepackage{framed}
\usepackage{paralist}

\usepackage{thmtools}
\usepackage{thm-restate}

\usepackage{algpseudocode,algorithm,algorithmicx}
\newcommand*\Let[2]{\State #1 $\gets$ #2}
\algrenewcommand\algorithmicrequire{\textbf{Precondition:}}
\algrenewcommand\algorithmicensure{\textbf{Postcondition:}}

\usepackage[dvipsnames]{xcolor}
\usepackage{tikz}
\usetikzlibrary{fit}
\usetikzlibrary{intersections}
\usetikzlibrary{decorations.pathreplacing}
\usetikzlibrary{calc}
\usepackage{tikz-3dplot}

\usepackage{amsmath}
\usepackage{amssymb}
\usepackage{xspace}
\usepackage{authblk}
\usepackage[color=green!20]{todonotes}

\newtheorem{claim}{Claim}

\newtheorem{problem}{Problem}

\def\cli{\textsc{Maximum Clique}\xspace}
\def\mis{\textsc{Maximum Independent Set}\xspace}

\tikzstyle{vp}=[circle,fill,inner sep=0pt, minimum size=0.1cm]
\tikzstyle{vps}=[circle,fill,inner sep=0pt, minimum size=0.065cm]

\def\centerarc[#1](#2)(#3:#4:#5)
    { \draw[#1] ($(#2)+({#5*cos(#3)},{#5*sin(#3)})$) arc (#3:#4:#5); }

\newcommand{\needle}[1]{\text{Needle}(#1)}
\newcommand{\fork}[3]{[#1 \leftarrow #2 \rightarrow #3]}
\newcommand{\forkb}[3]{[#1 \rightarrow #2 \leftarrow #3]}
\newcommand{\dir}[1]{\text{dir}(#1)}
\newcommand{\vcdim}{\text{VCdim}}
\newcommand{\iocp}{\text{iocp}}

\DeclareMathOperator{\ocp}{ocp}

\newcommand{\cl}{\mathcal X}

\title{EPTAS for Max Clique on Disks and Unit Balls}
\author[1]{Marthe Bonamy\textsuperscript{$\dagger$}}
\author[2]{Édouard Bonnet\footnote{This work was performed within the framework of the LABEX MILYON (ANR-10- LABX-0070) of Université de Lyon, within the program "Investissements d'Avenir" (ANR-11-IDEX-0007) operated by the French National Research Agency (ANR)}}
\author[3]{Nicolas Bousquet\footnote{This work has been supported by the ANR Project DISTANCIA (ANR-17-CE40-0015) operated by the French National Research Agency (ANR).}}
\author[2,4]{Pierre Charbit\textsuperscript{$\dagger$}}
\author[2,5]{Stéphan Thomassé}

\affil[1]{CNRS, LaBRI, Université de Bordeaux\\
  \texttt{marthe.bonamy@u-bordeaux.fr}}
\affil[2]{Université de Lyon (COMUE), \\CNRS, ENS de Lyon, Université Claude-Bernard Lyon 1,\\ LIP, Lyon, France\\
  \texttt{edouard.bonnet@dauphine.fr}, \texttt{stephan.thomasse@ens-lyon.fr}}
\affil[3]{CNRS, G-SCOP laboratory, Grenoble-INP, France\\
\texttt{nicolas.bousquet@grenoble-inp.fr}}
\affil[4]{Université Paris Diderot - IRIF, Paris, France\\
  \texttt{charbit@irif.fr}}
\affil[5]{Institut Universitaire de France}

\authorrunning{M. Bonamy, \'E. Bonnet, N. Bousquet, P. Charbit, S. Thomassé}

\begin{document}

\maketitle

\begin{abstract}
  We propose a polynomial-time algorithm which takes as input a finite set of points of $\mathbb R^3$ and computes, up to arbitrary precision, a maximum subset with diameter at most $1$.
  More precisely, we give the first randomized EPTAS and deterministic PTAS for \textsc{Maximum Clique} in unit ball graphs.
  Our approximation algorithm also works on disk graphs with arbitrary radii, in the plane.

  Almost three decades ago, an elegant polynomial-time algorithm was found for \textsc{Maximum Clique} on unit disk graphs [Clark, Colbourn, Johnson; Discrete Mathematics '90].
  Since then, it has been an intriguing open question whether or not tractability can be extended to general disk graphs.
  Recently, it was shown that the disjoint union of two odd cycles is never the complement of a disk graph [Bonnet, Giannopoulos, Kim, Rz{\k{a}}{\.{z}}ewski, Sikora; SoCG '18].
  This enabled the authors to derive a QPTAS and a subexponential algorithm for \textsc{Max Clique} on disk graphs.  
  In this paper, we improve the approximability to a randomized EPTAS (and a deterministic PTAS).
  More precisely, we obtain a randomized EPTAS for computing the independence number on graphs having no disjoint union of two odd cycles as an induced subgraph, bounded VC-dimension, and linear independence number.
  We then address the question of computing \textsc{Max Clique} for disks in higher dimensions.
  We show that intersection graphs of unit balls, like disk graphs, do not admit the complement of two odd cycles as an induced subgraph.
  This, in combination with the first result, straightforwardly yields a randomized EPTAS for \textsc{Max Clique} on unit ball graphs.
  In stark contrast, we show that on ball graphs and unit 4-dimensional disk graphs, \textsc{Max Clique} is NP-hard and does not admit an approximation scheme even in subexponential-time, unless the Exponential Time Hypothesis fails.
\end{abstract}

\section{Introduction}
In an \emph{intersection graph}, the vertices are geometric objects with an edge between any pair of intersecting objects. 
Intersection graphs have been studied for many different families of objects due to their practical applications and their rich structural properties~\cite{McKee1999, Brandstadt1999}.
Among the most studied ones are \emph{disk graphs}, which are intersection graphs of closed disks in the plane, and their special case, \emph{unit disk graphs}, where all the radii are equal.
Their applications range from sensor networks to map labeling~\cite{DBLP:conf/waoa/Fishkin03}, and many standard optimization problems have been studied on disk graphs, see for example~\cite{EJvL2009} and references therein.
Most of the hard optimization and decision problems remain NP-hard on disk graphs and even unit disk graphs.
For instance, disk graphs contain planar graphs \cite{koebe} on which several of those problems are intractable.
However, shifting techniques and separator theorems may often lead to subexponential classic or parameterized algorithms \cite{Alber04,Marx15,SmithWormald,Biro17}.
Many approximation algorithms have been designed specifically on (unit) disk graphs, or more generally on geometric intersection graphs, see for instance \cite{Chan03,Nieberg04,Nieberg05,Erlebach05,Leeuwen06,Gibson10} to cite only a few.
Besides ad hoc techniques, local search and VC-dimension play an important role in the approximability of problems on (unit) disk graphs.
For the main packing and covering problems (\mis, \textsc{Min Vertex Cover}, \textsc{Minimum Dominating Set}, \textsc{Minimum Hitting Set}, and their weighted variants) at least a PTAS is known.

However, all the techniques that we mentioned are only amenable to packing and covering problems.  
The \cli problem is arguably the most prominent problem which does not fall into those categories.
For example, anything along the lines of exploiting a small separator cannot work for \cli, where the densest instances are the hardest.
Therefore, it seems that new ideas are necessary to get improved approximate or exact algorithms for this problem.
This is why, in this paper, we focus on solving \cli on (unit) disk graphs in dimension 2 or higher.

\paragraph*{Previous results.}
In 1990, Clark \emph{et al.}~\cite{Clark90} gave an elegant polynomial-time algorithm for \cli on unit disk graphs when the input is a geometric representation of the graph.
It goes as follows: guess in quadratic time the two more distant centers of disks in a maximum clique (at distance at most $2$), remove all the centers that would contradict this maximality, observe that the resulting graph is co-bipartite.
Hence, one can find an optimum solution in polynomial time by looking for a maximum independent set in the complement graph, which is bipartite. 
However, recognizing unit disk graphs is NP-hard \cite{Breu98}, and even $\exists \mathbb{R}$-complete~\cite{Kang12}.
In particular, if the input is the mere unit disk graph, one cannot expect to efficiently compute a geometric representation in order to run the previous algorithm.
Raghavan and Spinrad showed how to overcome this issue and suggested a polynomial-time algorithm which does not require the geometric representation \cite{Raghavan03}. 
Their algorithm is a subtle \emph{blind} reinterpretation of the algorithm by Clark \emph{et al.}
It solves \cli on a superclass of the unit disk graphs or correctly claims that the input is not a unit disk graph.
Hence, it cannot be used to efficiently recognize unit disk graphs.

The complexity of \cli on general disk graphs is a notorious open question in computational geometry.
On the one hand, no polynomial-time algorithm is known, even when the geometric representation is given.
On the other hand, the NP-hardness of the problem has not been established, even when only the graph is given as input.

The piercing number of a collection of geometric objects is the minimum number of points that hit all the objects.
It is known since the fifties (although the first published records of that result came later in the eighties) that the piercing number of pairwise intersecting disks is $4$ \cite{stacho,danzer}.
An account of this story can be found in a recent paper by Har-Peled \emph{et al.} \cite{Harpeled18}.
Amb\"uhl and Wagner observed that this yields a $2$-approximation for \cli~\cite{Ambuhl05}.
Indeed, after guessing in polynomial time four points hitting a maximum clique and removing every disk not hit by those points, the instance is partitioned into four cliques; or equivalently, two co-bipartite graphs.
One can then solve optimally each instance formed by one co-bipartite graph and return the larger solution of the two. 
This cannot give a solution more than twice smaller than the optimum.
Since then, the problem has proved to be elusive with no new positive or negative results.
The question on the complexity and further approximability of \cli on general disk graphs is considered as folklore~\cite{bang2006}, but was also explicitly mentioned as an open problem by Fishkin~\cite{DBLP:conf/waoa/Fishkin03}, Amb\"uhl and Wagner~\cite{Ambuhl05}.
Cabello even asked if there is a 1.99-approximation for disk graphs with two sizes of radii~\cite{CabelloOpen,Cabello2015}.
Recently, Bonnet \emph{et al.} \cite{bonnetetal18} showed that the disjoint union of two odd cycles is not the complement of a disk graph.
From this result, they obtained a subexponential algorithm running in time $2^{\tilde{O}(n^{2/3})}$ for \cli on disk graphs, based on a win-win approach.
They also got a QPTAS by calling a PTAS for \mis on graphs with sublinear odd cycle packing number due to Bock \emph{et al.}~\cite{Bock14}, or branching on a low-degree vertex.


\paragraph*{Our results.}
Our main contributions are twofold.
The first is a randomized EPTAS (Efficient Polynomial-Time Approximation Scheme, that is, a PTAS in time $f(\varepsilon)n^{O(1)}$) for \mis on graphs of $\cl(d,\beta,1)$. The class $\cl(d,\beta,1)$ denotes the class of graphs whose neighborhood hypergraph has VC-dimension at most $d$, independence number at least $\beta n$, and no disjoint union of two odd cycles as an induced subgraph (for formal definitions see Section~\ref{sec:prelim}).

\begin{restatable}{theorem}{eptas}\label{thm:eptas}
  For any constants $d \in \mathbb N$, $0 < \beta \leqslant 1$, for every $0 < \varepsilon < 1$, there is a randomized $(1-\varepsilon)$-approximation algorithm running in time $2^{\tilde{O}({1/\varepsilon}^3)}n^{O(1)}$ for \mis on graphs of $\cl(d,\beta,1)$ with $n$ vertices.
\end{restatable}

Using the forbidden induced subgraph result of Bonnet \emph{et al.} \cite{bonnetetal18}, it is then easy to reduce \cli on disk graphs to \mis on $\cl(4,\beta,1)$ for some constant $\beta$.
We therefore obtain a randomized EPTAS (and a PTAS) for \cli on disk graphs, settling almost\footnote{The NP-hardness, ruling out a 1-approximation, is still to show.} completely the approximability of this problem.
\begin{restatable}{theorem}{disk}\label{thm:disks-eptas}
  There is a randomized EPTAS for \cli on disk graphs, even without geometric representation.
  Its running time is $2^{\tilde{O}(1/\varepsilon^3)}n^{O(1)}$ for a $(1-\varepsilon)$-approximation on a graph with $n$ vertices.
\end{restatable}

The second contribution is to show the same forbidden induced subgraph for unit ball graphs as the one obtained for disk graphs.
The proofs are radically different and the classes are incomparable.
So the fact that the same obstruction applies for disk graphs and unit ball graphs might be somewhat accidental.

\begin{restatable}{theorem}{forbidden}\label{thm:no-two-odd-cycles-ubg}
  A complement of a unit ball graph cannot have a disjoint union of two odd cycles as an induced subgraph.
  In other words, if $G$ is a unit ball graph, then $\iocp(\overline G) \leqslant 1$.
\end{restatable}
In the previous statement $\iocp$ denotes the \emph{induced odd cycle packing number} of a graph, i.e., the maximum number of odd cycles as a disjoint union in an induced subgraph.
Again, Theorem~\ref{thm:eptas} and Theorem~\ref{thm:no-two-odd-cycles-ubg} naturally lead to:
\begin{restatable}{theorem}{ball}\label{thm:ubg}
  There is a randomized EPTAS in time $2^{\tilde{O}({1/\varepsilon^3})}n^{O(1)}$ for \cli on unit ball graphs, even without the geometric representation.
\end{restatable}
Before that result, the best approximation factor was 2.553, due to Afshani and Chan~\cite{Afshani05}.
In particular, even getting a 2-approximation algorithm (as for disk graphs) was open.

Finally we show that such an approximation scheme, even in subexponential time, is unlikely for ball graphs (that is, 3-dimensional disk graphs with arbitrary radii), and unit 4-dimensional disk graphs.
Our lower bounds also imply NP-hardness.
To the best of our knowledge, the NP-hardness of \cli on unit $d$-dimensional disk graphs was only known when $d$ is superconstant ($d=\Omega(\log n)$)~\cite{Afshani08}.

In the following paragraphs, we sketch the principal lines of our two main contributions.

\paragraph*{EPTAS for \mis on $\cl(d,\beta,1)$}
The first main result of this paper asserts that if a graph $G$ satisfies that every two odd cycles are joined by an edge, the Vapnik-Chervonenkis dimension of the hypergraph of the neighborhoods of $G$ is bounded, and $\alpha(G)$ is at least a constant fraction of $|V(G)|$, then $\alpha(G)$ can be computed in polynomial time at any given precision.
More precisely, we present in that case a randomized EPTAS running in time $2^{\tilde{O}({1/\varepsilon^3})}n^{O(1)}$ and a deterministic PTAS.

Our algorithm works as follows.
We start by sampling a small subset of vertices.
Hoping that this small subset is entirely contained in a fixed optimum solution $I$, we include the selected vertices to our solution and remove their neighborhood from the graph.
Due to the classic result of Haussler and Welzl \cite{HausslerW86} on $\varepsilon$-nets of size $O(d/\varepsilon \log{1/\varepsilon})$ (where $d$ is the VC-dimension), this sampling lowers the degree in $I$ of the remaining vertices.
We compute a shortest odd cycle.
If this cycle is short, we can remove its neighborhood from the graph and solve optimally the problem in the resulting graph, which is bipartite by assumption.
If this cycle is long, we can efficiently find a small odd-cycle transversal.
This is shown by a careful analysis on the successive neighborhoods of the cycle, and the recurrent fact that this cycle is a shortest among the ones of odd length.

\paragraph*{The complement of the union of two odd cycles is not a unit ball graph}
Given a needle in $\mathbb R^3$ whose middlepoint is attached to the origin, one can apply a continuous motion in order to turn it around (a motion à la Kakeya, henceforth \emph{Kakeya motion}).
A Kakeya motion can be seen as a closed antipodal curve on the 2-sphere.
If we now consider two needles, each with a Kakeya motion, then the two needles have to go through a same position. 
This simply follows from the fact that two antipodal curves on the 2-sphere intersect.
The second main result of this paper is a translation of this Jordan-type theorem in terms of intersection graphs:
The complement of a unit ball graph does not contain the disjoint union of two odd cycles.
The proof can really be seen as two Kakeya motions, each one along the two odd cycles, leading to a contradiction when the needles achieve parallel directions.

Together with the first result, it implies a randomized EPTAS for \cli on disk graphs, and for the following problem: Given a set $S$ of points in $\mathbb R^3$, find a largest subset of $S$ of diameter at most $1$.

\paragraph*{Organization}
The rest of the paper is organized as follows.
In Section~\ref{sec:prelim}, we recall some relevant notations for graphs and elementary geometry, the definitions of VC-dimension, disk graphs, and approximation schemes.
We finish this section by introducing a class of graphs parameterized by three constants: the VC-dimension, the ratio \emph{independence number divided by number of vertices}, and the maximum number of odd cycles that can be found as a disjoint union in an induced subgraph.
In Section~\ref{sec:eptas}, we design a randomized EPTAS for \mis on this class.
In Section~\ref{sec:higher}, we show that complements of unit ball graphs do not have a disjoint union of two odd cycles as an induced subgraph.
This yields a randomized EPTAS for \cli on unit ball graphs (as well as on disk graphs).
This is tight in two directions: having different values of radii, and the dimension of the ambient space.
Indeed, we complement this positive result by showing that \cli is unlikely to even have a QPTAS on ball graphs where all the radii are arbitrarily close to 1, and on $4$-dimensional unit ball graphs.
In Section~\ref{sec:perspectives}, we make some observations about the EPTAS and propose some lines of thoughts on how to tackle the computational complexity of \cli in disk and unit ball graphs. 

\section{Preliminaries}\label{sec:prelim}

\paragraph*{Graph notations}
Let $G$ be a simple graph.
We denote by $\overline{G}$ its complement, i.e., the graph obtained by making every non-edge an edge and vice versa.
$V(G)$ and $E(G)$ represent its set of vertices and its set of edges, respectively.
We denote by $\alpha(G)$ the \emph{independence number} of $G$, i.e., the size of a maximum independent set (or stable set), and by $\omega(G)$ the \emph{clique number} of $G$, i.e., the size of a maximum clique.
For $S \subseteq V(G)$, its open neighborhood, denoted by $N_G(S)$, is the set of vertices that are not in $S$ and have a neighbor in $S$, and its closed neighborhood is defined by $N_{G}[S]=S\cup N_{G}(S)$. We omit the subscript $G$ if the graph is obvious from the context and we write $N_{G}(x)$ instead of $N_{G}(\{x\})$.

The \emph{odd cycle packing number} of $G$, denoted by $\ocp(G)$, is defined as the maximum number of vertex-disjoint odd cycles and the \emph{induced odd cycle packing number} of $G$, denoted by $\iocp(G)$, is the maximum number of vertex-disjoint odd cycles with no edge between any two of them.

\paragraph*{VC-dimension}
VC-dimension has been introduced by Vapnik and Chervonenkis in the seminal paper~\cite{vapnik}.
Let $H=(V,E)$ be a hypergraph.
A set $X$ of vertices of $H$ is \emph{shattered} if for every subset $Y$ of $X$ there exists a hyperedge $e \in E$ such that $e \cap X = Y$.
An intersection between $X$ and a hyperedge $e$ of $E$ is called a \emph{trace} (on $X$).
Equivalently, a set $X$ is shattered if all its $2^{|X|}$ traces exist.
The \emph{VC-dimension} of a hypergraph is the maximum size of a shattered set.
As an abuse of language, we call VC-dimension of a graph $G$, denoted by $\vcdim(G)$, the VC-dimension of the neighborhood hypergraph $(V(G),\{N_G(v) $ $|$ $v \in V(G)\})$.

\paragraph*{Geometric notations}
For a positive integer $d$, we denote by $\mathbb R^d$ the $d$-dimensional euclidean space.
If $x$ and $y$ are two points of $\mathbb R^d$, $xy$ is the straight-line segment whose endpoint are $x$ and $y$.
We denote by $d(x,y)$ the euclidean distance between $x$ and $y$.
A \emph{$d$-dimensional closed disk} is defined from a center $x \in \mathbb R^d$ and a radius $r \in \mathbb R^+$, as the set of points $\{y \in \mathbb R^d $ $|$ $d(x,y) \leqslant r\}$, i.e., at distance at most $r$ from $x$.
The \emph{diameter} of a subset $S \subseteq \mathbb R^d$ is defined as $\underset{x,y \in S}{\sup}d(x,y)$.
The \emph{piercing number} (also called \emph{hitting set} or \emph{transversal}) of a collection $\mathcal O$ of geometric objects in $\mathbb R^d$ is the minimum number of points of $\mathbb R^d$ that \emph{pierce} (or \emph{hit}) all the objects of $\mathcal O$, i.e., each object contains at least one of these points.

\paragraph*{Disk graphs and their forbidden induced subgraphs}

A \emph{$d$-dimensional disk graph} is the intersection graph of $d$-dimensional closed disks of $\mathbb R^d$.
We shorten $2$-dimensional disk graph in \emph{disk graph}, and $3$-dimensional disk graph in \emph{ball graph}.
A {$d$-dimensional unit disk graph} is the intersection graph of unit $d$-dimensional closed disks of $\mathbb R^d$, that is, disks with radius $1$.
Unit $d$-dimensional disk graphs can be thought of only with points: vertices are points (at the center of the disks) and two points are adjacent if they are at distance at most $2$.
In particular, solving \cli on those graphs is equivalent to finding a maximum sub-collection of points whose diameter is at most a fixed value.

Bonnet \emph{et al.} established the following obstruction for disk graphs.
\begin{theorem}[\cite{bonnetetal18}]\label{thm:disk-obstruction}
  The complement of a disk graph cannot have the disjoint union of two odd cycles as an induced subgraph.
\end{theorem}
This can be equivalently rephrased as: if $G$ is a disk graph, then $\iocp(\overline{G}) \leqslant 1$.

\paragraph*{Approximation schemes}
A PTAS (Polynomial-Time Approximation Scheme) for a minimization (resp. maximization) problem is an approximation algorithm which takes an additional parameter $\varepsilon>0$ and outputs in time $n^{f(\varepsilon)}$ a solution of value at most $(1+\varepsilon)\text{OPT}$ (resp. at least $(1-\varepsilon)\text{OPT}$) where $\text{OPT}$ is the optimum value.
Observe that from now on, we consider that approximation ratios of maximization problems are smaller than $1$, unlike the convention we used in the introduction. 
An EPTAS (Efficient PTAS) is the same with running time $f(\varepsilon)n^{O(1)}$, an FPTAS (Fully PTAS) is in time $\frac{1}{\varepsilon}^{O(1)}n^{O(1)}$, a QPTAS (Quasi PTAS) is in time $n^{\text{polylog}~n}$ for every $\varepsilon$.
Finally, and this is quite informal and not standard, we call SUBEXPAS (subexponential AS) an approximation scheme with running time $2^{n^{0.99}}$ for every $\varepsilon$.
All those approximation schemes can come deterministic or randomized.

\paragraph*{The class $\cl(d,\beta,i)$}
In the next section, we present a randomized EPTAS and a deterministic PTAS for approximating the independence number $\alpha$ on graphs with constant VC-dimension, linear independence number, and induced odd cycle packing number equals to 1. 

Actually, we extend these algorithms to the case $\iocp(G)=i$, for any constant $i$. Let $\cl(d,\beta,i)$ be the class of simple graphs $G$ satisfying:
\begin{itemize}
\item $\vcdim(G) \leqslant d$,
\item $\alpha(G) \geqslant \beta |V(G)|$, and
\item $\iocp(G) \leqslant i$.
\end{itemize}

For any positive constants $d, \beta<1, i$, we get a deterministic PTAS and a randomized EPTAS for \mis on $\cl(d,\beta,i)$.
  
\section{EPTAS for \mis on $\cl(d,\beta,i)$}\label{sec:eptas}

We start by showing that $\cl(d,\beta,1)$ has a randomized EPTAS.

\eptas*

\begin{proof}

Let $H$ be a graph in $\cl(d,\beta,1)$ with $n$ vertices and $I$ be a maximum independent set of $H$.
In particular, $|I| \geqslant \beta n$.
Since finding a maximum independent set in a bipartite graph can be done in polynomial time, we get the desired $(1-\varepsilon)$-approximation if we can find a vertex-set $T$ such that: 
\begin{itemize}
\item $T$ is an \emph{odd-cycle transversal}, i.e., its removal yields a bipartite graph, and 
\item $|T \cap I|\leqslant \varepsilon |I|$.
\end{itemize}

At high level, our algorithm will thus select and remove some odd-cycle transversals $T$, and then apply the bipartite case algorithm.
We will do this at least once for a set $T$ that satisfies the second item, with some strong guarantee.
Of course the key ingredient in finding a suitable odd-cycle transversal is the fact that $\iocp(H) \leqslant 1$.
Indeed, this implies that for any odd cycle $C$, the set $N[C]$ is an odd-cycle transversal.

Let $c := 8(\frac{1}{(\beta \varepsilon)^2}+\frac{1}{\beta \varepsilon}+1) = O(1/\varepsilon^2)$, $\delta := \frac{\varepsilon}{c}=O(\varepsilon^3)$, and $s:=\frac{10 d}{\delta} \log \frac{1}{\delta}$.
We call \emph{short odd cycle} an odd cycle of length at most $c$, and \emph{long odd cycle} an induced odd cycle of length more than~$c$.
First we can assume that $\beta n$ is larger than $2s$, otherwise we can find an optimum solution by brute-force in time $2^n=2^{\tilde{O}(1/\varepsilon^3)}$.
Hence, $|I| > 2s$.

\begin{claim}\label{clm:vc-dim}
There exists a subset $S\subseteq I$ of size $s=\frac{10 d}{\delta} \log \frac{1}{\delta}$ such that $N(S)$ contains all vertices that have more than $\delta|I|$ neighbors in $I$. 
\end{claim}

\begin{proof}
  Let $A$ denote the set of vertices $v$ such that $|N(v) \cap I| \geqslant \delta |I|$.
  We define the hypergraph $K:=(I,\{N(v) \cap I, v \in A\})$.
  By assumption on $H$, the hypergraph $K$ has VC-dimension $d$.
  By definition of $K$, all its edges have size at least $\delta |V(K)|$.
  A celebrated result in VC-dimension theory by Haussler and Welzl~\cite{HausslerW86}, later improved by Blumer \emph{et al.}~\cite{Blumer89}, ensures that every such hypergraph $K$ admits a hitting set (a set of vertices that intersects every edge) of size at most $\frac{10 d}{\delta} \log \frac{1}{\delta}$.
\end{proof}

Algorithmically, we have two ways of selecting the set $S$, leading to a deterministic PTAS or a randomized EPTAS.
Either we run the rest of the algorithm for every subset of $V(H)$ of size $\frac{10 d}{\delta} \log \frac{1}{\delta}$ inducing an independent set (which constitutes $n^{f(\varepsilon)}$ possible sets), or we use another result proven in \cite{HausslerW86}: not only the hitting set exists but a uniform sample of $V(K)$ of size $\frac{10 d}{\delta} \log \frac{1}{\delta}$ is a hitting set with high probability.
So we do the following $t := \lceil \frac{\log(10^{-10})}{\log(1-(\beta/2)^s)} \rceil = 2^{\tilde{O}(1/\varepsilon^3)}$ many times: we select uniformly at random a set $S$ of size $s=\frac{10 d}{\delta} \log \frac{1}{\delta}$, and continue the rest of the algorithm if $S$ is an independent set.
Since $|I| > 2s$, it holds that $\Pr(S \subseteq I) > (\beta/2)^s$.
As we try out $t$ samples, at least one satisfies $S \subseteq I$ with probability at least $1-(1-(\beta/2)^s)^t \geqslant 1-10^{-10}$.\\
  
We now assume that the sample $S$ satisfies the properties of Claim~\ref{clm:vc-dim}.
We start by putting in $T$ all the vertices of $N(S)$ (note that no such vertex is in $I$ since $I$ is an independent set).
We define the graph $H':=H-N(S)$.
We got rid of the vertices with at least $\delta |I|$ neighbors in $I$: in $H'$, there are no such vertices anymore.
We want to find an odd-cycle transversal in $H'$ that has few vertices in $I$. 

We now run a polynomial-time algorithm (see for instance \cite{AlonYZ97}) that determines whether the graph is bipartite and, if not, outputs a shortest odd cycle $C_{\text{og}}$ in $H'$.
\begin{itemize}
\item If $H'$ is bipartite, then $T:=N_H(S)$ is an odd-cycle transversal of $H$ with $|T \cap I|=0$.
\item If $g=|C_{\text{og}}| \leqslant c$, that is, if $C_{\text{og}}$ is a short odd cycle, then $|N_{H'}(C_{\text{og}}) \cap I| \leqslant c \delta |I| = \varepsilon|I|$, and therefore $T:=N_H(S)\cup N_{H'}[C_{\text{og}}]$ is an odd-cycle transversal of $H$ with $|T\cap I|\leqslant \varepsilon |I|$.
\end{itemize}

We can now safely assume that $g > c$, i.e., $C_{\text{og}}$ is a long odd cycle.
We decompose $H'$ into the successive neighborhoods of $C_{\text{og}}$, which we call \emph{layers}.
We define the first layer as $L_1 := N_{H'}(C_{\text{og}})$.
We define by induction the other layers as the non-empty sets $L_i := \{ v $ $|$ there exists $ u \in L_{i-1}$ with $uv \in E(H')$ and $v \notin L_j$ for $j<i \}$.
Let us denote by $\lambda$ the index of the last non-empty layer.
Before entering into the formal details of the second part of the proof let us briefly explain its structure:
\begin{itemize}
\item First, we observe that if there are many layers, there is one with index at most $\frac{2}{\beta \varepsilon}$ that contains at most $\frac{\varepsilon \beta}{2} n \leqslant \frac{\varepsilon}{2}|I|$ of the vertices. We can thus delete this layer, and note that connected components that do not contain $C_{\text{og}}$ are bipartite.
  We then focus on the component containing $C_{\text{og}}$, which has only few layers.
 \item Secondly, we show that this component admits an odd-cycle transversal of size at most $\frac{\varepsilon}{2} |I|$ (informally, the neighborhood at distance up to $O(\frac{1}{\varepsilon})$ of $O(\frac{1}{\varepsilon})$ consecutive vertices on the cycle $C_{\text{og}}$).
\end{itemize}

In other words, we can find $\varepsilon |I|$ vertices whose deletion yields a bipartite graph (see Figure~\ref{fig:large-successive-neighborhoods}), which together with $N(S)$ form the desired odd-cycle transversal.

\begin{figure}[!ht]
\centering
\begin{tikzpicture}[scale=0.5,vertex/.style={draw,circle,inner sep=-0.02cm}]
\def\s{41}
\def\l{6}
\def\ll{13}
\def\hd{0.2}
\def\vd{1}
\foreach \i in {1,...,\s}{
\node[vertex] (v\i) at (0,\i * \hd) {} ;
}
\foreach \i [count=\j from 1] in {2,...,\s}{
  \draw (v\i) -- (v\j) ;
}
\path (v1) edge[bend left=10] (v\s) ;
\foreach \j in {1,...,\ll}{
  \foreach \i in {1,...,\s}{
\draw (\vd * \j - \vd / 3, \i * \hd - \hd / 2) rectangle (\vd * \j + \vd / 3, \i * \hd + \hd / 2) ; 
}
}

\foreach \i in {8,...,16}{
  \node[vertex,fill=red] at (v\i) {} ; 
 \foreach \j in {1,...,\l}{
\draw[fill=red,opacity=0.5] (\vd * \j - \vd / 3, \i * \hd - \hd / 2) rectangle (\vd * \j + \vd / 3, \i * \hd + \hd / 2) ; 
}
}
\foreach \j in {7}{
  \foreach \i in {1,...,\s}{
\draw[fill=blue,opacity=0.5] (\vd * \j - \vd / 3, \i * \hd - \hd / 2) rectangle (\vd * \j + \vd / 3, \i * \hd + \hd / 2) ; 
}
}
\end{tikzpicture}
\caption{The layers (columns) and the strata (rows of a column).
  If the number of successive neighborhoods is large, a small cutset (in blue) is found among the first $\lceil \frac{2}{\beta \varepsilon} \rceil$ layers. To the right of this cutset, we know that the graph is bipartite. This brings us back to the case with fewer than $\frac{2}{\beta \varepsilon}$ layers, where we can find a small odd cycle transversal (in red).}
\label{fig:large-successive-neighborhoods}
\end{figure}
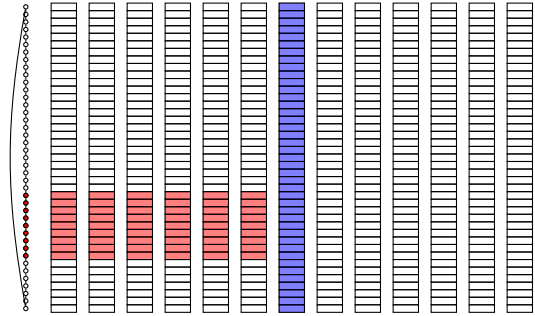

%

If $\lambda > \frac{2}{\beta \varepsilon}$, then there is some index $i \leqslant \lceil \frac{2}{\beta \varepsilon} \rceil$ such that $L_i$ is of size at most $\frac{\beta \varepsilon}{2} n \leqslant \frac{\varepsilon}{2} |I|$. 
We remove that layer $L_i$ from the graph.
Since $\iocp(H')=1$, the set $\bigcup_{i < j \leqslant \lambda}L_j$ induces a bipartite graph.
Indeed, it is disjoint from the closed neighborhood of the odd cycle $C_{\text{og}}$.
We can easily find a maximum independent set on this part of the graph, and focus on the other part, which is $C_{\text{og}} \cup \bigcup_{1 \leqslant j < i}L_j$.
We set $H'' := H'[C_{\text{og}} \cup \bigcup_{1 \leqslant j < i}L_j]$. If  $\lambda \leqslant \frac{2}{\beta \varepsilon}$, we set $H'':=H'$.

So the graph $H''$ has at most $\frac{2}{\beta \varepsilon}$ layers emanating from $C_{\text{og}}$.
We will find an odd-cycle transversal of size at most $\frac{\varepsilon}{2} |I|$.
We first need some new definitions.

For $1 \leqslant j \leqslant g$, let $S_j$ be the set of vertices $w \in V(H')$ such that there is a shortest path from $w$ to $C_{\text{og}}$ which ends in $v_j$, while no shortest path from $w$ to $C_{\text{og}}$ ends in $v_i$ with $i < j$ (note that $v_j \in S_j$).
We point out that the sets $(S_1,\ldots,S_g)$ induce a partition of each layer $L_k$.
This simply follows from the fact that for every vertex $w \in L_k$, there is a minimum index $j(w)$ such that there is a shortest path from $v$ to $C_{\text{og}}$ ending in $v_{j(w)}$.
For each pair $(k,\ell)$, we define a \emph{stratum} as $L_k^\ell := S_\ell \cap L_k$.
Note that if $L_k^\ell = \emptyset$, then for any $k'>k$, $L_{k'}^\ell = \emptyset$.

Let $z := \lceil \frac{4}{\beta \varepsilon} \rceil + 2$ and for any integer $\gamma$ such that $0 \leqslant \gamma \leqslant \frac{g}{z}-1$, let $S^\gamma := \underset{\gamma z+1 \leqslant j \leqslant (\gamma+1)z}{\bigcup}S_j$.
Informally, $S^\gamma$ consists of the layers emanating from $z$ consecutive vertices of $C_{\text{og}}$.
Note that if $\gamma \neq \gamma'$, then $S^{\gamma}$ and $S^{\gamma'}$ are disjoint.

\begin{claim}\label{clm:bip}
  For any non-negative integer $\gamma \leqslant \frac{2}{\beta \varepsilon}$, the graph $B := H''- S^\gamma$ is bipartite.
\end{claim}
\begin{proof}

  Observe that $\frac{g}{z}-1 \geqslant \frac{2}{\beta \varepsilon}$, so each $S^\gamma$ of the claim is well-defined.
  It holds that $C_{\text{og}} \cap S^\gamma = \bigcup_{\gamma z+1 \leqslant j \leqslant (\gamma+1)z}\{v_j\}$.
  We exhibit a proper $2$-coloring of $B$ by coloring its vertices as follows.
  We start by coloring each vertex of the path $C_{\text{og}} \setminus S^\gamma$ in an alternated fashion, i.e., one endpoint of the path gets color 0, its neighbor gets color 1, the next vertex gets color 0, and so on. 

  For each pair $(k,\ell)$ such that $1 \leqslant k < i$, $1 \leqslant \ell \leqslant g$, and $\ell \notin [\gamma z + 1, (\gamma+1)z]$, we color all the vertices in the stratum $L_k^\ell = S_\ell \cap L_k$ with the opposite color of the one used for the stratum $L_{k-1}^\ell = S_\ell \cap L_{k-1}$ 
  (with the convention that $L_0 := C_{\text{og}}$).
  This process colors unambiguously all the vertices of $B$.

  Let us prove that the resulting coloring is proper.
  First note that the vertices of a same stratum form an independent set.
  Indeed, assume by contradiction that $L_k^\ell$ contains an edge $uw$.
  There is a shortest path $P_1$ from $v_\ell$ to $u$ and a path $P_2$ from $w$ to $v_\ell$.
  Since $u$ and $w$ are in the same layer $L_k$, $P_1$ and $P_2$ have the same length; more precisely, $|P_1|=|P_2|=k < i \leqslant \frac{2}{\beta \varepsilon}$.
  Thus, $P_1, uw, P_2$ defines a closed walk with $2k+1$ edges.
  An odd closed walk of length $2k+1$ implies the existence of an odd induced cycle of length at most $2k+1$.
  As $2k+1 < \frac{4}{\beta \varepsilon}+1 < g$, we reach a contradiction on the minimality of $C_{\text{og}}$.
  
  There is no edge between a stratum $L_k^\ell$ and a stratum $L_{k'}^{\ell'}$ with $|k-k'| \geqslant 2$, by definition of the layers.
  Moreover, for $1 \leqslant \ell < \ell' \leqslant g$, there is no edge $uw$ with $u \in L_k^\ell$ and $w \in L_{k'}^{\ell'}$ with $\min(\ell'-\ell,\ell+g-\ell') \geqslant \frac{4}{\beta \varepsilon}+1$ since otherwise it would be possible to construct an odd cycle strictly shorter than $C_{\text{og}}$.
  Indeed if $P_1$ is a shortest path between $v_\ell$ and $u$ and $P_2$ is a shortest path between $w$ and $v_{\ell'}$.
  Then, $P_1, uw, P_2$ is a walk of length $k+k'+1$.
  However, a shortest path between $v_\ell$ and $v_{\ell'}$ within $C_{\text{og}}$ has length $\min(\ell'-\ell,\ell+g-\ell') \geqslant \frac{4}{\beta \varepsilon}+1 > k+k'+1$.
  Hence, the walk $P_1, uw, P_2$ can be extended into an odd closed walk of length strictly smaller than $g$, by taking the path from $v_{\ell'}$ to $v_\ell$ in $C_{\text{og}}$ with the same parity as $k+k'$; a contradiction.

  Therefore, if there is a monochromatic edge $uw$ in $B$, it must be between $L_k^\ell$ and $L_{k'}^{\ell'}$ with $|k-k'| \in \{0,1\}$, $\min(\ell'-\ell,\ell+g-\ell') < \frac{4}{\beta \varepsilon}+1$, and $\{\ell,\ell'\} \cap [\gamma z + 1, (\gamma+1) z] = \emptyset$.
  We fix $k, k', \ell, \ell'$ satisfying those conditions.
  We call \emph{small interval of $\ell$ and $\ell'$}, denoted by $\text{si}(\ell,\ell')$, the integer interval $[\ell,\ell']$ if $\min(\ell'-\ell,\ell+g-\ell')=\ell'-\ell$ and $[\ell',g] \cup [1,\ell]$ if $\min(\ell'-\ell,\ell+g-\ell')=\ell+g-\ell'$.
  What we showed in the previous paragraph implies that if $\{\ell,\ell'\} \cap [\gamma z + 1, (\gamma+1) z] = \emptyset$, then $\text{si}(\ell,\ell') \cap [\gamma z + 1, (\gamma+1) z] = \emptyset$.
  Indeed, the small interval of $\ell$ and $\ell'$ is a circular interval over $[1,g]$ of length less than $\frac{4}{\beta \varepsilon}+1 < z$.
  In particular, the vertices of $C_{\text{og}}$ indexed by the small interval of $\ell$ and $\ell'$ are all in $B$.
  
  Assume first that $k=k'$.
  There is a path $P_1$ from $v_\ell$ to $u$, and a path $P_2$ from $w$ to $v_{\ell'}$, both of length $k$.
  Since by assumption the color for $L_k^\ell$ is the same as the color for $L_{k'}^{\ell'}$, the vertices $v_\ell$ and $v_\ell'$ have the same color (by construction of the 2-coloring).
  Thus we have, in $C_{\text{og}}-S^\gamma$, a path $P$ indexed by $\text{si}(\ell,\ell')$ from $v_{\ell'}$ to $v_\ell$ of even length less than $\frac{4}{\beta \varepsilon}+1$.
  We emphasize that it is crucial that $\text{si}(\ell,\ell') \cap [\gamma z + 1, (\gamma+1) z] = \emptyset$ (meaning that all the vertices of $C_{\text{og}}$ indexed by $\text{si}(\ell,\ell')$ are still in $B$), to deduce that there is a path of even length between $v_\ell$ and $v_{\ell'}$. It follows from the mere fact that $v_\ell$ and $v_\ell'$ have the same color.   
  Finally, the concatenation $P_1, uw, P_2, P$ yields an odd cycle of odd length less than $2k+1+\frac{4}{\beta \varepsilon}+1 \leqslant \frac{8}{\beta \varepsilon}+2 < g$.

  Now let us assume that $|k-k'|=1$; say, without loss of generality, $k'=k+1$.
  In that case, by construction of the 2-coloring, the edge can only be monochromatic if $v_\ell$ and $v_{\ell'}$ receive distinct colors.
  Furthermore, there is a path $P_1$ from $v_\ell$ to $u$, and a path $P_2$ from $w$ to $v_{\ell'}$ with length of distinct parities ($k$ and $k+1$, respectively).
  Moreover, since $v_\ell$ and $v_{\ell'}$ get distinct colors, there is in $C_{\text{og}}-S^\gamma$, a path $P$ indexed by $\text{si}(\ell,\ell')$ from $v_{\ell'}$ to $v_\ell$ of odd length at most $\frac{4}{\beta \varepsilon}+1$.
  Again, we crucially use that all the vertices of $C_{\text{og}}$ indexed by $\text{si}(\ell,\ell')$ are in $B$, to deduce that $P$ is of odd length from the fact that $v_\ell$ and $v_{\ell'}$ get distinct colors. 
  Finally, the concatenation $P_1, uw, P_2, P$ is an odd cycle of length less than $k+1+k+1+\frac{4}{\beta \varepsilon}+1 \leqslant \frac{8}{\beta \varepsilon}+3 < g$; a contradiction.

We conclude that the $2$-coloring is indeed proper.
\end{proof}

\begin{algorithm}
  \caption{EPTAS for \mis on $\cl(d,\beta,1)$}
    \label{alg:eptas}
  \begin{algorithmic}[1]
    \Require{$H$ satisfies $d := \vcdim(G)=O(1)$, $\alpha(G) \geqslant \beta |V(G)|$, and $\iocp(G) \leqslant 1$}
    \Function{Stable}{$H,\varepsilon$}:
    \Let{$c$}{$8(1/(\beta \varepsilon)^2+1/(\beta \varepsilon)+1)$}
    \Let{$\delta$}{$\frac{\varepsilon}{c}$}
    \Let{$s$}{$\frac{10d}{\delta} \log{\frac{1}{\delta}}$}
    \If{$\beta |V(H)| < 2s$}
       solve $H$ optimally by brute-force                          \Comment{$|V(H)| = \tilde{O}(1/\varepsilon^3)$}
    \EndIf
      \For{$\_ \gets 1 \textrm{ to } t=2^{\tilde O(1/\varepsilon^3)}$}{
        \Let{$S$}{\text{uniform sample of $V(G)$ of size $s$}}  \Comment{$S \subseteq I$ with probability $> (\frac{\beta}{2})^s$}
         \If{$G[S]$ contains an edge} break and go to the next iteration \EndIf  
         \Let{$H'$}{$H-N[S]$}                                    \Comment{remove $S$ and its neighborhood}  
         \Let{$C_{\text{og}}$}{shortest odd cycle in $H'$}       \Comment{in polynomial time \cite{AlonYZ97}}
         \Let{$g$}{$|C_{\text{og}}|$}                            \Comment{$C_{\text{og}}=v_1v_2 \ldots v_g$}
         \If{$g \leqslant c$}                  \Comment{short odd cycle}
         \Let{$S$}{$S~\cup~$ max stable on the bipartite $H'-N[C_{\text{og}}]$} \Comment{$\iocp(G)=1$}
         \Else{}                                              \Comment{long odd cycle}
           \Let{$L_\ell$}{vertices of $H'$ at distance exactly $\ell$ from $C_{\text{og}}$}   
           \Let{$L_i$}{smallest layer among $\{L_\ell\}_{1 \leqslant \ell \leqslant \lceil 2/{\beta \varepsilon} \rceil}$} \Comment{$|L_i| \leqslant \frac{\varepsilon}{2}\alpha(G)$}
           \Let{$H''$}{$H'[C_{\text{og}} \cup \bigcup_{1 \leqslant j < i}L_j]$}   
           \Let{$S_k$}{vertices of $H''$ whose closest vertex on $C_{\text{og}}$ of minimum index is $v_k$}   
           \Let{$z$}{$\lceil \frac{4}{\beta \varepsilon} \rceil + 2$}
           \Let{$S^\gamma$}{smallest set among $\{\bigcup_{k \in [\gamma z+1,(\gamma + 1)z]}S_k\}_{\gamma \in [0,\lfloor \frac{2}{\beta \varepsilon} \rfloor]}$}                \Comment{$|S^\gamma| \leqslant \frac{\varepsilon}{2}\alpha(G)$}
           \Let{$S$}{$S~\cup~$ max stable on the bipartite $H'[\bigcup_{j > i}L_j]$}    \Comment{$\iocp(G)=1$}
           \Let{$S$}{$S~\cup~$ max stable on the bipartite $H''-S^\gamma$}              \Comment{Claim~\ref{clm:bip}}
         \EndIf
      }
      \EndFor
      \State \Return{$S$ at the iteration maximizing its cardinality}
      \EndFunction
   \Ensure{output $S$ is a stable set of size at least $(1-\varepsilon)\alpha(G)$ with high probability}
  \end{algorithmic}
\end{algorithm}

Since the sets of $\{S^\gamma\}_{\gamma \in [0, \lfloor \frac{2}{\beta \varepsilon} \rfloor]}$ are pairwise disjoint, a smallest set of the collection satisfies $|S^\gamma| \leqslant \frac{\beta \varepsilon}{2}n \leqslant \frac{\varepsilon}{2} |I|$.
By Claim~\ref{clm:bip}, removing this $S^\gamma$ from $H''$ makes the graph bipartite.
We finally compute a maximum independent set in polynomial time in $H''- S^\gamma$.
We return the best solution found.
The pseudo-code is detailed in Algorithm~\ref{alg:eptas}.
\end{proof}

\disk*

\begin{proof}
  With a geometric representation, we could invoke the following argument to get a linear maximum stable set. 
  Recall that the piercing number of a family of geometric objects is the minimum number of points such that each object contains at least one of those points.
  The piercing number of a collection of pairwise intersecting disks in the plane is $4$ \cite{stacho,danzer}.
  The number of faces in an arrangement of $n$ circles (disk boundaries) is $O(n^2)$, and all the points within one face hits the same disks.
  In time $O(n^8)$, one can therefore exhaustively guess four points piercing a maximum clique $\mathcal C$.
  We can remove all the disks which are not hit by any of those four points, since they are not part of $\mathcal C$.
  This new instance $G$ can have its vertices partitioned into four cliques, hence $\alpha(\overline{G}) \geqslant |V(G)|/4$.

  Without a geometric representation, we suggest the following.
  Disk graphs are $6\omega$-degenerate (and closed by induced subgraphs), i.e., there is a vertex of degree at most $6\omega$, where $\omega$ is the clique number of the graph.
  Furthermore, the neighborhood of this vertex can easily be partitioned into $6$ cliques.
  We can find such a vertex in polynomial time.
  We branch on two outcomes.
  Either this vertex is in a maximum clique: we run the approximation of Theorem~\ref{thm:eptas} on its closed neighborhood $G$ which satisfies $\alpha(\overline G) \geqslant |V(G)|/6$.
  Or this vertex is not in any maximum clique: we delete it from the graph.
  Our branching tree has size $2n+1$, so it only costs an extra linear multiplicative factor. 

  The VC-dimension of the neighborhoods of disk graphs, and even pseudo-disk graphs \cite{Aronov18}, is at most~$4$.
  Since the VC-dimension of a graph is equal to the one of its complement, the VC-dimension of $\overline G$ is also at most $4$.
  Finally, by Theorem~\ref{thm:disk-obstruction} \cite{bonnetetal18}, $\iocp(\overline{G}) \leqslant 1$.
  We only call the approximation algorithm (a polynomial number of times) with disk graphs $G$ such that $\overline{G} \in \cl(4,\frac{1}{6},1)$ (argument without the geometric representation) or $\overline{G} \in \cl(4,\frac{1}{4},1)$ (argument with the geometric representation).
  Hence, we conclude by Theorem~\ref{thm:eptas}.
\end{proof}

It is known that the VC-dimension of unit balls and the piercing number of pairwise intersecting unit balls are both constant.
In Section~\ref{sec:higher}, we will show that the induced odd cycle packing number of unit ball graphs is at most $1$.
This is the missing element for the EPTAS to also work for \cli on unit balls.

We conclude this section by extending the EPTAS to work for constant (not necessarily~$1$) induced odd cycle packing number.

\begin{theorem}\label{thm:eptas-bis}
  For any constants $d, i \in \mathbb N$, $0 < \beta \leqslant 1$, for every $\varepsilon > 0$, there is a randomized $(1-\varepsilon)$-approximation algorithm running in time $2^{\tilde{O}({1/\varepsilon}^3)}n^{O(1)}$ for \mis on graphs of $\cl(d,\beta,i)$ with $n$ vertices.
\end{theorem}
\begin{proof}
  Let $I$ be a maximum independent set.
  We show by induction on $i$ that we can find in time $2^{\tilde{O}({1/\varepsilon'}^3)}n^{O(1)}$ a stable set of size $(1-i\varepsilon')|I|$.
  The base case is Theorem~\ref{thm:eptas}.
  We assume that there is such an algorithm when the induced odd cycle packing number is $i-1$.
  We follow Algorithm~\ref{alg:eptas} with a graph $H$ such that $\iocp(H)=i$.
  On line 15 and 23, the graph is not necessarily bipartite anymore (on line 24, the resulting graph is still bipartite in this case).
  Although, the induced odd cycle packing number is decreased to $i-1$.
  So, by the induction hypothesis we get a stable set within a factor $(1-(i-1)\varepsilon')$ of the optimum.
  To get there, we removed a subset of vertices of size at most $\varepsilon'|I|$.
  Therefore, the solution $S$ that we obtain satisfies $|S| \geqslant (1-i\varepsilon')|I|$.

  We obtain the theorem by setting $\varepsilon := i \varepsilon'$ since $i$ is absorbed in the $\tilde{O}$ in the running time.
\end{proof}

\section{Balls and higher dimensions}\label{sec:higher}

We start by showing the main result of this section: for every unit ball graph $G$, $\iocp(\overline G) \leqslant 1$.
A \emph{closed polygonal chain} $C$ in $\mathbb R^d$ is defined by a set of points (or vertices) $x_1, x_2, \ldots, x_p \in \mathbb R^d$ as the straight-edge segments $x_1x_2$, $x_2x_3$, \dots, $x_{p-1}x_p$, $x_px_1$.
We call \emph{direction} of a non-zero vector its equivalence class by the relation $\vec{u} \sim \vec{v} \Leftrightarrow \exists \lambda \in \mathbb R^+, \vec{u} = \lambda \vec{v}$.
We denote the direction of $\vec{u}$ by $\dir{\vec{u}}$.
We define the set of directions $$\needle{C}:=\underset{1 \leqslant i \leqslant p}{\bigcup} \fork{x_{i-1}}{x_i}{x_{i+1}} \cup \forkb{x_{i-1}}{x_i}{x_{i+1}},$$ 
where the indices are taken modulo $p$, $\fork{x_{i-1}}{x_i}{x_{i+1}} := \{\dir{\overrightarrow{x_ix}} ~|~ x \in x_{i-1}x_{i+1} \},$ and $\forkb{x_{i-1}}{x_i}{x_{i+1}} := \{\dir{\overrightarrow{xx_i}} ~|~ x \in x_{i-1}x_{i+1} \}$. 

\begin{lemma}\label{lem:same-direction}
  Let $C_1$ and $C_2$ be two closed polygonal chains of $\mathbb R^3$ on an odd number of vertices each.
  Then, $\needle{C_1} \cap \needle{C_2}$ is non-empty.
\end{lemma}

\begin{proof}
  We want to establish the existence of a direction which is common to $\needle{C_1}$ and $\needle{C_2}$.
  For that, we yield a continuous map from $\mathbb R^+$ to $\needle{C_1}$ that we naturally interpret as a union of curves on the 2-sphere (unit sphere in $\mathbb{R}^3$).
  Indeed the set of directions in $\mathbb R^3$ is isomorphic to the set of points on the 2-sphere.
  As we do not need a specific parameterization of the curve, we just describe how we continuously move a vector $\overrightarrow{ab}$ whose direction runs through the entire set $\needle{C_1}$.
  Let $x_1, x_2, \ldots, x_p$ be the vertices of $C_1$.

  Let us first prove that $\needle{C}$ is a closed curve when $C$ is an odd cycle. To do so we prove that we can continuously modify the current vector $\dir{\overrightarrow{x_1x_2}}$ into any other vector in $\needle{C}$.
  We start with $a$ in $x_1$ and $b$ in $x_2$.
  We continuously move $a$ from $x_1$ to $x_3$ (on the straight-edge segment $x_1x_3$) while $b$ stays fixed at $x_2$ 
 (we are in moving in $\fork{x_1}{x_2}{x_3}$).
  For the next step, $a$ is fixed at $x_3$ and $b$ continuously moves from $x_2$ to $x_4$.
  (we are in moving in $\forkb{x_{2}}{x_3}{x_{4}}$).
  And in general, we move the point with index $i-1$ from $x_{i-1}$ to $x_{i+1}$ while the other point stays fixed at $x_i$ (where the indices are modulo $p$).
  Since $p$ is odd, we reach the situation where $b$ is set to $x_1$ and $a$ is set to $x_2$ when we have completed once the walk on the closed polygonal chain.  We can repeat a walk on the chain once again, so that $a$ is back to $x_1$ and $b$ is back to $x_2$, and we stop.
  One may observe that this process spans $\needle{C_1}$.
  
  This defines a closed curve $\mathcal C_1$ on the 2-sphere since we are finally back to $\dir{\overrightarrow{x_1x_2}}$ from where we started.
  Furthermore, $\mathcal C_1$ is \emph{antipodal}, i.e., closed by taking antipodal points.
  Indeed, for each direction attained in a $\fork{x_{i-1}}{x_i}{x_{i+1}}$, we reach the opposite direction in $\forkb{x_{i-1}}{x_i}{x_{i+1}}$. 
  Similarly $\needle{C_2}$ draws a closed antipodal curve on the 2-sphere $\mathcal C_2$.
  The curves $\mathcal C_1$ and $\mathcal C_2$ intersect since they are closed and antipodal.
  An intersection point corresponds to a direction shared by $\needle{C_1}$ and $\needle{C_2}$.
\end{proof}

We will apply this lemma on the closed polygonal chains $C_1$ and $C_2$ formed by the centers of unit balls realizing two odd cycle complements.
The contradiction will come from the fact that not all the pairs of centers $x \in C_1$ and $y \in C_2$ can be at distance at most 2.

\forbidden*
\begin{proof}
  Let $x_1, x_2, \ldots, x_p \in \mathbb R^3$ (resp.~$y_1, y_2, \ldots, y_q \in \mathbb R^3$) be the centers of unit balls representing the complement of an odd cycle of length $p$ (resp. $q$), such that $x_i$ and $x_{i+1}$ (resp. $y_i$ and $y_{i+1}$) encode the non-adjacent pairs.
  Let $C_1$ (resp. $C_2$) be the closed polygonal chain obtained from the centers $x_1, x_2, \ldots, x_p$ (resp. $y_1, y_2, \ldots, y_q$).
  By Lemma~\ref{lem:same-direction}, there are two collinear vectors $\overrightarrow{x_ix}$ and $\overrightarrow{y_jy}$ with $x$ on the straight-line segment $x_{i-1}x_{i+1}$ and $y$ on the straight-line segment $y_{j-1}y_{j+1}$.

  \begin{figure}[!ht]
    \centering
    \begin{tikzpicture}[scale=1.5,dot/.style={fill,circle,inner sep=-0.03cm}]
      \coordinate (a) at (-0.095,0,0) ;
      \coordinate (b) at (-0.2,1.6,0.5) ;
      \coordinate (c) at (0,1.6,-0.5) ;
      \coordinate (x) at (-0.05,1.6,-0.25) ;

      \coordinate (d) at (-0.28,0,0) ;
      \coordinate (e) at (0.8,2,0.8) ;
      \coordinate (f) at (-0.4,2,-0.4) ;
      \coordinate (y) at (-0.2,2,-0.2) ;

      \foreach \i in {a,b,c,d,e,f,x,y}{
        \node[dot] at (\i) {} ;
      }

      \draw[dashed,thin,red] (b) -- (a) -- (c) ;
      \draw[dashed,thin,blue] (e) -- (d) -- (f) ;
      \draw[dotted] (b) -- (c) ;
      \draw[dotted] (e) -- (f) ;
      \draw[thick,red] (a) -- (x) ;
      \draw[thick,blue] (d) -- (y) ;

      \node at (-0.3,-0.2,0) {$x_i$} ;
      \node at (-0.17,2.13,-0.2) {$x$};
      \node at (0,1.75,-0.2) {$y$} ;
      \node at (-0.08,-0.2,0) {$y_j$} ;

      \begin{scope}[xshift=3cm]
      \coordinate (x) at (-0.095,0,0) ;
      \coordinate (b) at (-0.2,0.3,0.5) ;
      \coordinate (c) at (0,-0.3,-0.5) ;
      \coordinate (a) at (-0.05,1.6,-0.25) ;

      \coordinate (d) at (-0.28,0,0) ;
      \coordinate (e) at (0.8,2,0.8) ;
      \coordinate (f) at (-0.4,2,-0.4) ;
      \coordinate (y) at (-0.2,2,-0.2) ;

      \foreach \i in {a,b,c,d,e,f,x,y}{
        \node[dot] at (\i) {} ;
      }

      \draw[dashed,thin,red] (b) -- (a) -- (c) ;
      \draw[dashed,thin,blue] (e) -- (d) -- (f) ;
      \draw[dotted] (b) -- (c) ;
      \draw[dotted] (e) -- (f) ;
      \draw[thick,red] (a) -- (x) ;
      \draw[thick,blue] (d) -- (y) ;

      \node at (-0.3,-0.2,0) {$x_i$} ;
      \node at (-0.17,2.13,-0.2) {$x$};
      \node at (0,1.75,-0.2) {$y_j$} ;
      \node at (-0.08,-0.2,0) {$y$} ;
      \end{scope}
    \end{tikzpicture}
    \caption{The two cases for the collinear vectors $\overrightarrow{x_ix}$ and $\overrightarrow{y_jy}$.}
    \label{fig:two-cases-collinear}
  \end{figure}
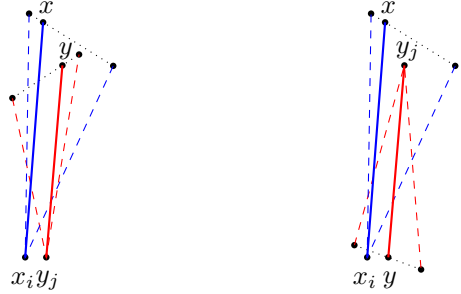
  Let us suppose that $\overrightarrow{x_ix}$ and $\overrightarrow{y_jy}$ have the same direction (Figure~\ref{fig:two-cases-collinear}, left).
  In the plane\footnote{or \emph{a} plane if it is not unique} containing $x_i$, $x$, $y_j$, and $y$, those four points are in convex position and the convex hull is cyclically ordered $x_i,x,y,y_j$.
  We obtain a contradiction by showing that the sum of the diagonals $d(x_i,y)+d(y_j,x)$ is strictly smaller than the sum of two opposite sides $d(x_i,x)+d(y_j,y)$.
  
  Considering edges and non-edges, for every $i$, the only vertices at distance at least $2$ from $y_i$ in $\cup \{ x_j , y_j\}$ are $y_{i-1}$ and $y_{i+1}$. Then, we have $d(x_i,y_{j-1}) \leqslant 2 < d(y_j,y_{j-1})$ and $d(x_i,y_{j+1}) \leqslant 2 < d(y_j,y_{j+1})$.
  The points strictly closer to $x_i$ than to $y_j$ form an open half-space.
  In particular, they form a convex set and all the points in the segment $y_{j-1}y_{j+1}$ are therefore closer to $x_i$ than to $y_j$.
  Hence, $d(x_i,y) < d(y_j,y)$.
  Symmetrically, $d(y_j,x) < d(x_i,x)$.
  So $d(x_i,y)+d(y_j,x) < d(y_j,y)+d(x_i,x)$, a contradiction.
  
  Let us now assume that $\overrightarrow{x_ix}$ and $\overrightarrow{y_jy}$ have opposite direction (Figure~\ref{fig:two-cases-collinear}, right).
  In that case, the four coplanar points $x_i,x,y_j,y$ are in convex position in their plane and the convex hull is cyclically ordered $x_i,x,y_j,y$.
  We will attain the similar contradiction that $d(x,y)+d(x_i,y_j) < d(x_i,x)+d(y_j,y)$.
  As previously, $d(y_{j-1},x_{i-1}) \leqslant 2 < d(x_i,x_{i-1})$ and $d(y_{j-1},x_{i+1}) \leqslant 2 < d(x_i,x_{i+1})$.
  Hence, by convexity, we have $d(y_{j-1},x) < d(x_i,x)$.
  Similarly, $d(y_{j+1},x) < d(x_i,x)$.
  We obtained that point $x$ is closer to $y_{j-1}$ and $y_{j+1}$ than to $x_i$.
  Therefore, applying again the convexity argument, we get that $d(x,y) < d(x,x_i)$.
  Besides, $d(x_i,y_j) \leqslant 2 < d(y,y_j)$.
  So we arrive at the contradiction $d(x,y)+d(x_i,y_j) < d(x,x_i)+d(y,y_j)$.
\end{proof}

It is folklore that unit ball graphs have VC-dimension $4$.
At the price of a multiplicative factor $n$ in the running time, one can guess a vertex $v$ in a maximum clique of a unit ball graph $G$, and look for a clique in its neighborhood $H := G[N(v)]$.
It can be easily shown that the neighborhood of this vertex (in fact, of any vertex) can be partitioned into $25$ cliques, following the proof that the kissing number for unit spheres is $12$.
Thus $\alpha(\overline H) \geqslant |V(H)|/25$.
Therefore, from Theorem~\ref{thm:eptas} and Theorem~\ref{thm:no-two-odd-cycles-ubg}, we immediately obtain the following.

\ball*

From Theorem~\ref{thm:no-two-odd-cycles-ubg}, we also get an exact subexponential algorithm running in time $2^{\tilde{O}(n^{2/3})}$ for \cli on unit ball graphs.
Indeed, such an algorithm was obtained for the class of graphs $G$ such that $\iocp(\overline G) \leqslant 1$ \cite{bonnetetal18}.

In sharp contrast, we will prove that with an extra dimension or with different radii (even arbitrarily close to each other), a PTAS is highly unlikely.
In both cases, we show that there is a constant ratio which is not attainable even in subexponential time, unless the ETH fails.
A \emph{2-subdivision} of a simple graph is the graph obtained by subdividing each edge exactly twice.
A \emph{co-2-subdivision} is the complement graph of the 2-subdivision of a graph.
The class of the 2-subdivisions (resp. co-2-subdivisions) is the set of all graphs that are the 2-subdivisions (resp.~co-2-subdivisions) of a simple graph. 

The following result was shown by Bonnet \emph{et al}.
\begin{theorem}[\cite{bonnetetal18}]
  There is a constant $\alpha > 1$ such that \cli on co-2-subdivisions is not $\alpha$-approximable even in time $2^{n^{0.99}}$, unless the ETH fails.
  Moreover, \cli is NP-hard on co-2-subdivisions.
\end{theorem}
Therefore, we just need to show that all co-2-subdivisions can be realized by our geometric objects. 
This appears like a simple and powerful method to rule out a PTAS (QPTAS, and even SUBEXPAS) for a geometric clique problem.

The co-2-subdivision of a graph with $n$ vertices and $m$ edges can be thought of as follows.
It is made of a clique on $n$ vertices, representing the initial vertices, and a clique on $2m$ vertices minus a perfect matching, representing endpoints of the initial edges.
Each anti-matched pair of vertices corresponds to an edge in the initial graph.
Each vertex representing one endpoint of an initial edge is adjacent to all the vertices representing the initial vertices but this endpoint. 

\begin{theorem}\label{thm:4udg}
  The class of 4-dimensional unit disks contains all the co-2-subdivisions.
\end{theorem}
\begin{proof}
  Given any simple graph $G=(V,E)$ with $n$ vertices and $m$ edges, we want to build a set $S$ of $n+2m$ points in $\mathbb R^4$ where each vertex $v$ is represented by a point $p(v)$ and each edge $e$, by two points $p^+(e)$ and $p^-(e)$.
  A pair of points in $S$ should be at distance at most 2, except if it is the two points of the same edge $e$, or if it is a point $p(v)$ with either a $p^+(vw)$ or a $p^-(wv)$; those pairs should be at distance strictly more than 2.
  We call $x,y,z,t$ the coordinates of $\mathbb R^4$.
  Let $\mathcal P$ be the plane defined by the intersection of the hyperplanes of equation $z=0$ and $t=0$.
  The projection $\pi$ onto $\mathcal P$ of the points $p^+(e)$ (resp. $p^-(e)$) fall regularly on the "top" part (resp. "bottom" part) of a circle $\mathcal C$ of $\mathcal P$ centered at the origin $O=(0,0,0,0)$ and of diameter $2$, such that for each edge $e$, $\pi(p^+(e))$ and $\pi(p^-(e))$ are antipodal on $\mathcal C$.
  We just defined the points $\pi(p^+(e))$ and $\pi(p^-(e))$.
  The actual points $p^+(e)$ and $p^-(e)$ will be fixed later by moving them very slightly away from their projection in a two-dimensional plane orthogonal to $\mathcal P$. 
  Let $\eta$ be the maximum distance between a pair $(\pi(p^+(e)),\pi(p^-(e')))$ with $e \neq e'$.
  By construction $\eta<2$.
  
  Let $\mathcal P^\bot$ be the (2-dimensional) plane containing the center $O$ of $\mathcal C$ and orthogonal to $\mathcal P$; in other words, the intersection of the hyperplanes of equation $x=0$ and $y=0$.
  We observe that all the points of $\mathcal P^\bot$ are equidistant to all the points $\pi(p^+(e))$ and $\pi(p^-(e))$.
  We place all the points $p(v)$ in $\mathcal P^\bot$ regularly spaced on a arc of a circle lying on $\mathcal P^\bot$ centered at $O$ and of radius $\sqrt{3}-\varepsilon$.
  One can notice that for any $(v, e, s) \in V \times E \times \{+,-\}$, $d(p(v),\pi(p^s(e))) = \sqrt{4-2 \sqrt 3 \varepsilon + \varepsilon^2}$.
  We will choose $\varepsilon \ll 2 - \eta \ll 1$ so that this shared distance is just below 2 and the points $\{p^+(e),p^-(e)\}_{e \in E}$ realize the same adjacencies than their projection by $\pi$.

  For every $e=uv \in E$, we place $p^+(e)$ such that $\overrightarrow{\pi(p^+(e))p^+(e)}=-\frac{(\varepsilon+\varepsilon')}{\lVert Op(u) \rVert}\overrightarrow{Op(u)}$ and $p^-(e)$ such that $\overrightarrow{\pi(p^-(e))p^-(e)}=-\frac{(\varepsilon+\varepsilon')}{\lVert Op(v) \rVert}\overrightarrow{Op(v)}$.
  In words, we push very slightly $p^+(e)$ (resp. $p^-(e)$) away from $\pi(p^+(e))$ (resp. $\pi(p^-(e))$) in the opposite direction of $\overrightarrow{Op(u)}$ (resp. $\overrightarrow{Op(v)}$).
  This way, $p^+(e)$ is at distance more than 2 from $p(u)$.
  Indeed, $d(p^+(e),p(u))=\lVert \overrightarrow{p^+(e)\pi(p^+(e))} + \overrightarrow{\pi(p^+(e))O} + \overrightarrow{Op(u)} \rVert = \lVert \overrightarrow{\pi(p^+(e))O} + \overrightarrow{Op(u)} + \overrightarrow{p^+(e)\pi(p^+(e))} \rVert = \lVert \overrightarrow{\pi(p^+(e))O} + (\sqrt 3 -\varepsilon+\varepsilon+\varepsilon')\frac{\overrightarrow{Op(u)}}{\lVert \overrightarrow{Op(u)}\rVert} \rVert = \sqrt{1+(\sqrt 3+\varepsilon')^2} > 2$.
  Similarly, $d(p^-(e),p(v))>2$.
  We choose $\varepsilon'$ infinitesimal (in particular, $\varepsilon' \ll \varepsilon$) so that, still, $d(p^+(e),p(w))<2$ for any $w \neq u$ and $d(p^-(e),p(w))<2$ for any $w \neq v$.
  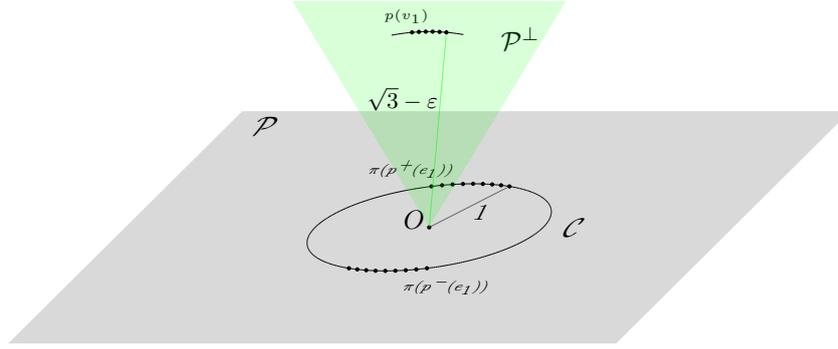
\begin{figure}
    \centering
    \begin{tikzpicture}[
        dot/.style={fill,circle,inner sep=-0.02cm},
      ]
      \def\xs{0.5}
      \coordinate (a) at (-4,0,-4) {} ;
      \coordinate (b) at (4,0,-4) {} ;
      \coordinate (c) at (4,0,4) {} ;
      \coordinate (d) at (-4,0,4) {} ;
      \fill[opacity=0.15] (a) -- (b) -- (c) -- (d) -- cycle ;
      \coordinate (O) at (0,0,0) {} ;
      \node[dot] at (O) {} ;
      \node at (-0.2,0.1,0) {$O$} ;
      \node[xslant=\xs]  at (-3.5,0,-3.5) {$\mathcal P$} ;
      \node[xslant=\xs]  at (1.9,0,0) {$\mathcal C$} ;
      
      \def \r{1.5}
      \foreach \i in {70,75,80,85,90,95,100,105,110}{
        \pgfmathsetmacro{\CosValue}{\r * cos(\i)}
        \pgfmathsetmacro{\SinValue}{\r * sin(\i)}
        \node[dot] (pm\i) at (\CosValue,0,\SinValue) {};
        \node[dot] (pp\i) at (\CosValue,0,-\SinValue) {};
      }
      \foreach \i [count=\j from 2] in {1,...,360}{
        \pgfmathsetmacro{\CosValue}{\r * cos(\i)}
        \pgfmathsetmacro{\SinValue}{\r * sin(\i)}
        \pgfmathsetmacro{\CosValueb}{\r * cos(\j)}
        \pgfmathsetmacro{\SinValueb}{\r * sin(\j)}
        \draw[very thin] (\CosValue,0,\SinValue) -- (\CosValueb,0,\SinValueb) {};
      }
      \node[xslant=\xs] (pp1) at (-1,0,-2) {\tiny{$\pi(p^+(e_1))$}} ;
      \node[xslant=\xs] (pm1) at (1,0,2) {\tiny{$\pi(p^-(e_1))$}} ;
      \draw[opacity=0.5] (pp70) -- (O) ;
      \node[xslant=\xs] at (0.5,0,-0.5) {\footnotesize{1}} ;

      \coordinate (e) at (-1.8,3,0) {} ;
      \coordinate (f) at (1.8,3,0) {} ;
      \node at (1.2,2.5,0) {$\mathcal P^{\bot}$} ;
      \node at (-0.3,2.8,0) {\tiny{$p(v_1)$}} ;

      \fill[green,opacity=0.15] (O) -- (e) -- (f) -- cycle ;

      \foreach \i in {85,87,89,91,93,95}{
        \pgfmathsetmacro{\CosValue}{1.73 * \r * cos(\i)}
        \pgfmathsetmacro{\SinValue}{1.73 * \r * sin(\i)}
        \node[dot] (q\i) at (\CosValue,\SinValue,0) {};
      }
      \draw[green,opacity=0.5] (q85) -- (O) ;
      \node at (-0.35,1.7,0) {\footnotesize{$\sqrt{3}-\varepsilon$}} ;
      
      \foreach \i [count=\j from 81] in {80,...,100}{
        \pgfmathsetmacro{\CosValue}{1.73 *\r * cos(\i)}
        \pgfmathsetmacro{\SinValue}{1.73 *\r * sin(\i)}
        \pgfmathsetmacro{\CosValueb}{1.73 *\r * cos(\j)}
        \pgfmathsetmacro{\SinValueb}{1.73 *\r * sin(\j)}
        \draw[very thin] (\CosValue,\SinValue,0) -- (\CosValueb,\SinValueb,0) {};
      }
    \end{tikzpicture}
    \caption{The overall construction for 4-dimensional unit disks. We only represent the centers.}
    \label{fig:4dim-unit-disks}
  \end{figure}
\end{proof}

\begin{corollary}
  There is a constant $\alpha > 1$ such that \cli on 4-dimensional unit disks is not $\alpha$-approximable even in time $2^{n^{0.99}}$, unless the ETH fails.
  Moreover, \cli is NP-hard on 4-dimensional unit disks.
\end{corollary}

The proof of Theorem~\ref{thm:4udg} can be tweaked for balls of different radii.
For a real $\varepsilon > 0$, we say that the radii of a representation (or the representation itself) are \emph{$\varepsilon$-close} if the radii are all contained in the interval $[1,1+\varepsilon]$.
We denote by $\mathcal B(1,1+\varepsilon)$ the ball graphs having an $\varepsilon$-close representation.

\begin{theorem}\label{thm:bg}
 For any $\varepsilon > 0$, the subclass of ball graphs $\mathcal B(1,1+\varepsilon)$ contains all the co-2-subdivisions.
\end{theorem}
\begin{proof}
  Let $x, y, z$ be the coordinates of $\mathbb R^3$.
  We start similarly and define the same $\pi(p^+(e))$ and $\pi(p^-(e))$ as in the previous construction for the points $p^+(e)$ and $p^-(e)$ on a circle $\mathcal C$ of diameter $2$ centered at $O=(0,0,0)$ on a plane $\mathcal P$ of equation $z=0$.
  One difference is that $\pi(p^+(e))$ and $\pi(p^-(e))$ are no longer projections.
  We then place the points $p(v)$ regularly spaced along the $z$-axis.
  More precisely, if the vertices are numbered $v_1,v_2,\ldots,v_n$, the position of $p(v_i)$ is $(0,0,\sqrt 3+i\varepsilon')$.
  The radius of the disk representing $v_i$ centered at $p(v_i)$ is set to $r_i := \sqrt{1+(\sqrt 3+i\varepsilon')^2}-1+\varepsilon''$.
  The radii associated to the centers $p^+(e)$ and $p^-(e)$ are all set to 1.
  We move $p^+(v_iv_j)$ (resp. $p^-(v_iv_j)$) away from $\pi(p^+(v_iv_j))$ (resp. $\pi(p^-(v_iv_j))$) in the opposite direction of $\overrightarrow{p^+(v_iv_j)p(v_i)}$ (resp. $\overrightarrow{p^-(v_iv_j)p(v_j)}$) by an infinitesimal quantity, in order to only suppress the overlap of the disks centered at $p^+(v_iv_j)$ and $p(v_i)$ (resp. $p^-(v_iv_j)$ and $p(v_j)$).
  We make $\varepsilon'$ and $\varepsilon''$ small enough that all the values $r_i$ are between 1 and $1+\varepsilon$.
\end{proof}

We call \emph{quasi unit ball graphs} those graphs in the intersection $\bigcap_{\varepsilon > 0}\mathcal B(1,1+\varepsilon)$.
As a corollary, we get some strong inapproximability even for quasi unit ball graphs.

\begin{corollary}
  There is a constant $\alpha > 1$ such that \cli on quasi unit ball graphs is not $\alpha$-approximable even in time $2^{n^{0.99}}$, unless the ETH fails.
  Moreover, \cli is NP-hard on quasi unit ball graphs.
\end{corollary}

\section{Remarks and further directions}\label{sec:perspectives}


  The algorithm of Theorem~\ref{thm:eptas} also works for weighted graphs.
  The slight modifications to approximate \textsc{Maximum Weighted Independent Set} consist in sampling $S$ \emph{proportionally to the weights}, and to remove the \emph{lightest} layer $L_i$ (among the first $\lceil {2}/{\beta \varepsilon} \rceil$) and the \emph{lightest} set $S^\gamma$ (rather than the ones of minimum size).
  We then use repeatedly that \textsc{Maximum Weighted Independent Set} can be solved in polynomial time on bipartite graphs.

  This implies a randomized EPTAS for \textsc{Maximum Weighted Clique} on disk graphs and unit ball graphs, with the same arguments used to get $|I| \geqslant \beta |V(G)|$.
  Now, what we obtain is $w(I) \geqslant \beta w(V(G))$ where $w$ is the weight function and $w(X):=\Sigma_{u \in X}w(u)$.

  \medskip
  One might wonder what is the constant hidden in $O(1)$ in the time complexity of the randomized EPTAS $f(\varepsilon)n^{O(1)}$.
  Here is how to achieve near quadratic time $f(\varepsilon) n^2 \log n$ where $n$ is the number of vertices of our unit ball graph $G$ (even without a geometric representation).
  The first observation is that instead of finding an optimum stable set in a bipartite graph (which we do several times as a subroutine), one only needs a $1+\varepsilon$-approximation of it.
  This can be done in time $f(\varepsilon)m$, where $m$ is the number of edges \cite{Duan14}.
  We also need to overcome our first branching guess of a vertex which belongs to an optimal solution (and hence multiply by $n$ our complexity).
  To achieve this, we start by packing greedily disjoint neighborhoods $N(v_1),N(v_2), \ldots, N(v_k)$ of $G$, while it is possible.
  Then we consider the set of subgraphs $G_1,G_2,...,G_k$ induced by the vertices at distance at most 3 of $v_1,v_2,...,v_k$, respectively.
Observe that by maximality of the packing, every vertex is at distance at most 2 of some $v_i$, and thus, every edge, and even every clique of $G$ belongs to at least one $G_i$.
By a volume argument, any vertex, and thus any edge, belongs to at most a constant number of graphs $G_i$.
Thus we only need to compute the maximum clique over the graphs $G_i$, each with number of edges $m_i$, with the property that $\sum_{i=1}^km_i=O(m)$.
Now, the difference is that the clique number of each $G_i$ is at least a constant fraction of its number of vertices, so the sampling step succeeds with probability at least a positive constant.
We consider the complement of each graph $G_i$ and approximate the independence number.
The main and only obstacle is that finding a shortest odd cycle in quadratic time seems hopeless.
Indeed, the other elements of the algorithm: removing $N[S]$, computing the sets $L_i$ and $S^\gamma$, removing the lightest of them, and $1+\varepsilon$-approximating the maximum stable set in a bipartite graph, can all be done in quadratic time.

We only compute one (not necessarily shorter) odd cycle of length $h$, via breadth-first search for instance.
We can assume that $h =\omega(1/\varepsilon^2)$ otherwise we are done.
We take potentially \emph{thicker} slices for $S^\gamma$: instead of $z=\Theta(1/\varepsilon)$ consecutive strata, we take $10 \beta \varepsilon h$ consecutive strata.
Either the suggested $2$-coloring of $H''-S^\gamma$ is proper and we are done.
Or there is a monochromatic edge, which, going through the different cases of Claim~\ref{clm:bip}, yields a new odd cycle of length $O(1/\varepsilon)$ (which is an easy case to handle) or short-cutting the previous odd cycle by at least $\Omega(\varepsilon h)$ vertices.
In the latter case, our new odd cycle is shorter by a constant multiplicative factor $1-\varepsilon$.
Hence, after at most $O(\log n)$ improvements, we find an odd cycle which is small enough to conclude.
Therefore, the overall complexity of the algorithm is $f(\varepsilon)n^2 \log n$.

  
  \medskip

  The obvious remaining question is the complexity of \cli in disk graphs and in unit ball graphs.
  An interesting direction would be to find a toy problem on which we could prove NP-hardness.
  A nice class, which appears to be a subclass of unit ball graphs, is that of the so-called \emph{Borsuk graphs}: We are given some (small) real $\varepsilon>0$ and a finite collection $V$ of unit vectors in $\mathbb R^3$. The Borsuk graph $B(V,\varepsilon)$ has vertex set $V$ and its edges are all pairs $\{v,v'\}$ whose dot product is at most $-1+\varepsilon$ (i.e. near antipodal).

  The difficulty of computing the (weighted) independence number on Borsuk graphs is also an open question.
  One can go one simplification step further and focus on a subclass of Borsuk graphs, namely quadrangulations of the projective plane.
  These well-studied objects have the striking property to be either bipartite or 4-chromatic.
  The complexity of \textsc{Maximum Weighted Independent Set} is open for 4-chromatic quadrangulations.
  Finally, the last simplification step in this direction is the following.

\begin{problem}
What is the computational complexity of \textsc{Maximum Weighted Independent Set} on the \emph{moebius grid}?
\end{problem}

By \emph{moebius grid} we mean a grid with an even number of columns in which the first and last columns are antipodally identified. 
Note that, in this example, every two odd cycles intersect ($\text{ocp}=1$).
However, there is no small odd-cycle transversal.
If the size of the grid is $\sqrt n \times \sqrt n$, then the smallest odd-cycle transversal has size roughly $\sqrt n$.

\medskip

A second question is to derandomize the EPTAS.
  The difficulty here is concentrated in the sampling.
  The VC-dimension argument seems easy to deal with, however we need that our sampling falls in the maximum independent set (or at least in some independent set which is close to maximum). 

  \medskip

  Another natural question is to find a superclass of geometric intersection graphs which both contain unit balls and disks.
  More generally, is it possible to explain why we have the same forbidden induced subgraph (the complement of a disjoint union of two odd cycles) for disk graphs and unit ball graphs?
  So far, the proofs that this is indeed an obstruction (see Theorem 6 in \cite{bonnetetal18} and Theorem~\ref{thm:no-two-odd-cycles-ubg} of this paper) are quite different for the two classes.
  
  Let us call \emph{quasi unit disk graphs} those disk graphs that can be realized for any $\varepsilon > 0$ with disks having all the radii in the interval $[1,1+\varepsilon]$.
  Recall that we showed that, for the clique problem, quasi unit ball graphs are unlikely to have a QPTAS, while unit ball graphs admit an EPTAS.
  In dimension 2, it can be easily shown that unit disk graphs form a proper subset of quasi unit disk graphs, which form themselves a proper subset of disk graphs.
  Can we find for this intermediate class an efficient exact algorithm solving \cli?
  \begin{problem}
   Is there a polynomial-time algorithm for \cli on quasi unit disk graphs?
  \end{problem}


  \medskip

  Our PTAS works for \mis under three hypotheses.
  While it is clear that we crucially need that $\text{iocp} \geqslant 1$ (or at least that $\text{iocp}$ is constant), as far as we can tell, the boundedness of the VC-dimension and the fact that the solution is fairly large might not be required.
  \begin{problem}
    Is there a PTAS for \mis on graphs without the union of two odd cycles as an induced subgraph?
  \end{problem}

  Atminas and Zamaraev \cite{Atminas16} showed that the complement of $K_2+C_s$ is not a unit disk graph when $s$ is odd (where $K_2$ is an edge and $C_s$ is a cycle on $s$ vertices).
  Is this obstruction enough to obtain an alternative polynomial-time algorithm for \cli on unit disk graphs? 
  \begin{problem}
  Is \mis solvable in polynomial-time on graphs excluding the union of an edge and an odd cycle as an induced subgraph?
  \end{problem}

\medskip

\bibliographystyle{abbrv}


\end{document}